\documentclass[10pt,nocopyrightspace]{sigplanconf}
\usepackage{amsthm}
\usepackage{amsmath}
\usepackage{amssymb}
\usepackage{comment}
\usepackage{verbatim}
\usepackage{graphicx}
\usepackage{framed}
\usepackage{tikz}
\usepackage{rotating}
\usepackage{varwidth}
\usepackage{pdflscape}
\usepackage{bm}
\usepackage{mathrsfs}
\usepackage{bussproofs}
\usepackage{subfigure}
\usepackage{xspace}

\renewcommand\hat{}

\newcommand{\rom}[1]{\uppercase\expandafter{\romannumeral #1\relax}}

\newcommand{\turnstile}{\bm{\vdash}}

\newcommand{\I}{\Rightarrow}
\newcommand{\LI}{\multimap}
\newcommand{\A}{\forall}
\newcommand{\E}{\exists}

\newcommand{\OT}{\otimes}
\newcommand{\OP}{\oplus}
\newcommand{\SQ}{\turnstile}
\newcommand{\AMPOP}{\;\&\;}
\newcommand{\AMPCON}{\&}
\newcommand{\Chi}{X}
\newcommand{\omicron}{o}
\newcommand{\assign}{\leftarrow}
\newcommand{\deref}[1]{\hat{*}#1}
\newcommand{\app}{{\ }}

\theoremstyle{plain}
\newtheorem{thm}{Theorem}
\newtheorem{lem}[thm]{Lemma}
\theoremstyle{definition}
\newtheorem{defn}[thm]{Definition}

\newenvironment{scprooftree}[1]%
  {\gdef\scalefactor{#1}\begin{center}\proofSkipAmount \leavevmode}%
  {\scalebox{\scalefactor}{\DisplayProof}\proofSkipAmount \end{center} }

\newenvironment{customprooftree}
{\begin{scprooftree}{.75}}
{\end{scprooftree}}

\newcommand{\evalrelation}[4]{\langle #1, #2 \rangle \leadsto (#3,#4)}
\newcommand{\naturals}{\mathbb{N}}
\newcommand{\programs}{\mathbb{R}}
\newcommand{\memoryupdate}[3]{#1[#2\mapsto #3]}
\newcommand{\soswhile}[2]{while\ #1\ do\ #2}

\newcommand{\sosruleone}[2]{
\AxiomC{#1}
\UnaryInfC{#2}
}
\newcommand{\sosruletwo}[3]{
\AxiomC{#1}
\AxiomC{#2}
\BinaryInfC{#3}
}
\newcommand{\sosbinop}[1]{
\AxiomC{$\evalrelation{E_1}{M}{N_1}{M'}$}
\AxiomC{$\evalrelation{E_2}{M'}{N_2}{M''}$}
\BinaryInfC{$\evalrelation{E_1 \hat{#1} E_2}{M}{N_1 #1 N_2}{M''}$}
}

\newcommand{\sosgreaterthan}[2]{
\AxiomC{$\evalrelation{E_1}{M}{N_1}{M'}$}
\AxiomC{$\evalrelation{E_2}{M'}{N_2}{M''}$}
\AxiomC{#1}
\TrinaryInfC{$\evalrelation{E_1 \hat{>} E_2}{M}{#2}{M''}$}
}

\newcommand{\sosrulethree}[4]{
\AxiomC{#1}
\AxiomC{#2}
\AxiomC{#3}
\TrinaryInfC{#4}
}
\newcommand{\sosrulefour}[5]{
\AxiomC{#1}
\AxiomC{#2}
\AxiomC{#3}
\AxiomC{#4}
\QuaternaryInfC{#5}
}
\newcommand{\threeproofs}[3]{
\begin{customprooftree}
                #1
                #2
                #3
        \alwaysNoLine
        \TrinaryInfC{}
\end{customprooftree}
}
\newcommand{\twoproofs}[2]{
\begin{customprooftree}
                #1
                #2
        \alwaysNoLine
        \BinaryInfC{}
\end{customprooftree}
}
\newcommand{\oneproof}[1]{
\begin{customprooftree}
#1
\end{customprooftree}
}
\newcommand{\sequent}[3]{#1;#2\turnstile#3}

\newcommand{\usequent}[2]{\sequent{\Gamma}{#1}{#2}}

\newcommand{\ubsequent}[1]{\usequent{\Delta}{#1}}

\newcommand{\termmappingname}{t_{\programs}}

\newcommand{\memmappingname}{t_{m}}

\newcommand{\evaltname}{e}

\newcommand{\memformname}{m}

\newcommand{\termmapping}[1]{\termmappingname(#1)}

\newcommand{\NT}[1]{#1} %
\newcommand{\memmapping}[1]{\memmappingname(#1)}

\newcommand{\vt}[1]{#1}%
\newcommand{\addt}[2]{(add \;#1 \;#2)}
\newcommand{\subt}[2]{(sub \;#1 \;#2)}
\newcommand{\gett}[1]{(get \;#1)}
\newcommand{\sett}[2]{(set \;#1 \;#2)}
\newcommand{\seqt}[2]{(seq \;#1\linebreak \;#2)}
\newcommand{\gtt}[2]{(gt \;#1 \;#2)}
\newcommand{\wht}[2]{(wh \;#1 \;#2)}

\newcommand{\evalt}[3]{(\evaltname\;#1\;#2\;#3)}

\newcommand{\memform}[2]{(\memformname\;#1\;#2)}

\begin{document}

\doi{nnnnnnn.nnnnnnn}

\title{Towards Reasoning About Properties of Imperative Programs using
  Linear Logic}
\authorinfo{Daniel DaCosta}{University of Minnesota}{dacosta@cs.umn.edu}

\maketitle

\begin{abstract}
  In this paper we propose an approach to reasoning about properties of
imperative programs.  We assume in this context that the meanings of
program constructs are described using rules in the natural semantics
style with the additional observation that these rules may involve the
treatment of state.  Our approach involves modeling natural semantics
style rules within a logic and then reasoning about the behavior of
particular programs by reasoning about proofs in that logic. A key
aspect of our proposal is to use a fragment of linear logic called
Lolli (invented by Hodas and Miller) to model natural semantics style
descriptions.  Being based on linear logic, Lolli can provide logical
expression to resources such as state.  Lolli additionally possesses
proof-theoretic properties that allow it to encode natural semantics
style descriptions in such a way that proofs in Lolli mimic the
structure of derivations based on the natural semantics rules.  We
will discuss these properties of Lolli and demonstrate how they can be
exploited in modeling the semantics of imperative programs and in
reasoning about such models.
\end{abstract}

\section{Introduction}
\label{sec:introduction}

This paper concerns an approach to reasoning about the properties of
imperative programs.
Such programs, written in languages like Java and C, play an important
role in safety- and security critical systems.
They are pervasive, for example, in the software contained in
medical devices and financial systems.
Programs that malfunction in such contexts can lead to catastrophic
system behavior.
The underlying motivation for this work is that through the process of
formal reasoning we can establish the absence of such bugs before
these programs are run and thereby preclude undesirable behavior
after their deployment.

Our objective in this work is not to reason about properties of
particular programs but, rather, to develop a broad
framework within which such reasoning may be conducted.
An important ingredient of such a framework is a logic for
describing the semantics of the programming language in which programs
are constructed; a formalization of the semantics can then be combined
with the description of a given program to model its overall
behavior.
An aspect that needs special treatment when dealing with
imperative programs in this setting is the notion of state: imperative
programs typically manipulate memory by storing and looking up values
in relevant cells and how exactly they do this is important to
understanding their behavior.
Thus, the logic that we choose for our framework must facilitate the
description as well as the analysis of the role of
state in computations.

In constructing the framework we desire, we must also choose an
approach to presenting the semantics of a programming language.
We propose in this work to use the natural semantics style introduced
by Kahn~\cite{kahn87stacs} for this purpose.
Natural semantics style allows the meaning of a programming language
construct to be modeled via derivations that closely reflect the
actual computations that result from the construct.
Thus, the process of reasoning about program behavior boils down
naturally to reasoning about natural semantics style derivations.
In our framework, programming language semantics will be modeled by
translating these natural semantics descriptions into the underlying
logic.
This actually places two further constraints on the logic.
First, it should have a structure that supports a natural encoding of
natural semantics style descriptions.
Second, the inference process in the logic should correspond
transparently to the process of constructing natural semantics style
derivations; this property allows reasoning about natural semantics style
derivations to be reduced uniformly to reasoning about proofs in the
logic.

The main thrust of the work in this paper is to identify a logic that
satisfies the constraints described above and that would thereby be a
suitable choice for encoding programs and programming language
semantics within the framework we seek to design.
We contend that linear logic, a logic of resources and
actions invented by Jean-Yves Girard~\cite{girard87tcs}, provides a
natural means for treating state-based aspects of computation and
hence constitutes a good starting point.
However, the logic we use needs also to allow for a treatment of 
natural semantics style descriptions.
We argue that Lolli, a fragment of linear logic identified by Hodas
and Miller \cite{hodas91lics}, has such a character.
To provide substance to our claim, we demonstrate how this logic can
be used to model the semantics of a small collection of constructs in
an imperative programming language.
We further show how the meta-theoretic properties of Lolli allow us to
translate an informal style of reasoning based on natural semantics
derivations into reasoning about derivations in Lolli.
Although a formalization of this reasoning process is beyond the scope
of this work, we believe that this can be done following the approach
used in the Abella system~\cite{gacek08ijcar}.

The rest of this paper is structured as follows.
In the next section, we describe a simple imperative programming
language and we present the meanings of the constructs in this
language in natural semantics style.
We then consider a small program in this language and show how we can
organize the process of reasoning about it around natural semantics style
derivations.
This language and reasoning example provide us the means for
explaining and defending our main contributions in the sections that
follow.
In Section \ref{sec:speclogic} we present Lolli and we discuss the
properties of derivations in it.
In the following section we show how Lolli can be used to formalize
the imperative programming language described earlier.
In Section~\ref{sec:reasoning}, we demonstrate how the properties of
Lolli can be used in reasoning.
In particular, we show that the informal reasoning process based on
the natural semantics style presentation translates naturally into
reasoning about derivations in Lolli.
We conclude the paper in Section~\ref{sec:conclusion} with
indications of directions in this work that may be worthwhile to
pursue in the future.

\section{A Simple Imperative Programming Language}
\label{sec:imp}

In this section we define an imperative programming language and its
evaluation semantics.  We use these definitions to demonstrate how a
program written in the language can be reasoned about informally.

\subsection{The Language}
\label{subsec:evalsems}

The syntax of programs in our imperative programming language ${\cal
  L}$ is given by the following rules:
\begin{align*}
R\mathrm{::=}&\;\;\naturals\;|\; R+R \;|\; R-R \;|\; R>R \;|\; \deref{R}  \\
&|\; R\assign{}R \;|\; R;R \;|\; (R) \;|\; while\; R \; do\; R\\
\end{align*}
In this definition, the symbol $\naturals$ represents the category of
expressions corresponding to the non-negative
integers and the programs $\deref{}R,R_1\assign{}R_2,$ and $R_1;R_2$
correspond to memory lookup (like in \texttt{C}), update of the
value stored at $R_2$ to $R_1$, and evaluation of $R_1$ followed by
evaluation of $R_2$, respectively. The last construct included by the
syntax rules permits indefinite iteration in the language.

We will need a model of memory in order to present the semantics of
the constructs in our language. Towards this end, we will represent
memory as a partial function from natural numbers (denoted also in an
overloaded fashion by $\naturals$) to expressions in the language that
correspond to natural numbers.
Given memory $M$, we will use the notation $\memoryupdate{M}{x}{y}$ to
denote a modified memory given by the following partial function:
\begin{equation*}
\memoryupdate{M}{x}{y}(a) =
\left\{
\begin{array}{ll}
  y  & \mbox{if } x = a  \\
  M(a) & \mbox{if } x \neq a
\end{array}
\right.
\end{equation*}
Notice that although memory is modelled as a function from the
conceptual domain of natural numbers, we will often want to do a
``lookup'' using the result of a computation. By an abuse of notation,
we will allow memory to be ``applied'' to the expressions that denote
natural numbers in the language.

We present the semantics of the constructs in our language by
explaining what it means to evaluate them. We do this by defining an
``evaluation relation'' that we write as
\[ \evalrelation{R}{M}{N}{M'}.  \] This relation is to be read as
``the program $R$ evaluated in memory $M$ returns value $N$ and
modifies the memory to $M'$'' or, when memory is not of interest,
``$R$ evaluates to value $N$.'' When referring to components of this
relation we may refer to the ``input program expression'', the ``input
memory'', the ``return value'', and the ``output memory'',
respectively.  We define this relation through rules in the natural
semantics style that are presented in Figure \ref{fig:sosimp}. For the
uninitiated reader, each of these rules is to be read as asserting
that the relation shown below the line holds if all the relations or
properties above the line hold; the former is called the conclusion of
the rule and the latter are called its premises. Notice that if a rule
has no premises then its conclusion is unconditionally true.

\begin{figure*}
\begin{framed}
\begin{center}

\threeproofs
   {
   \sosruleone{}{$\evalrelation{N}{M}{N}{M}$}}
   {\sosbinop{+}}{\sosbinop{-}}

\twoproofs
   {\sosgreaterthan{$N_1>N_2$}{$1$}}
   {\sosgreaterthan{$N_1\leq N_2$}{$0$}}

\twoproofs
   {\sosruletwo{$\evalrelation{R}{M}{N}{M'}$}{$M'(N)=N'$}
               {$\evalrelation{\deref{R}}{M}{N'}{M'}$}}
   {\sosrulethree{$\evalrelation{R_1}{M}{N_1}{M'}$}
                 {$\evalrelation{R_2}{M'}{N_2}{M''}$}
                 {$M''(N_1)=N_3$}
                 {$\evalrelation{R_1\assign{}R_2}{M}{N_2}
                                {\memoryupdate{M''}{N_1}{N_2}}$}}

\oneproof
   {\sosruletwo{$\evalrelation{R_1}{M}{N_1}{M'}$}
               {$\evalrelation{R_2}{M'}{N_2}{M''}$}
               {$\evalrelation{R_1;R_2}{M}{N_2}{M''}$}}

\twoproofs
   {\sosruleone{$\evalrelation{R_1}{M}{0}{M'}$}
               {$\evalrelation{\soswhile{R_1}{R_2}}{M}{0}{M'}$}}
   {\sosrulefour{$\evalrelation{R_1}{M}{N_1}{M'}$}
                {$\evalrelation{R_2}{M'}{N_2}{M''}$}
                {$\evalrelation{\soswhile{R_1}{R_2}}{M''}{N_3}{M'''}$}
                {$N_1>0$}
                {$\evalrelation{\soswhile{R_1}{R_2}}{M}{N_3}{M'''}$}}
\end{center}
\caption{
 Evaluation semantics for the imperative language $\cal L$
}
\label{fig:sosimp}
\end{framed}
\end{figure*}

A few comments are in order with regard to the rules in
Figure~\ref{fig:sosimp}. First, these rules are meant to be
read as schemata: actual rules are to be generated by instantiating
the  schema variables $N,N_1,N_2,$ and $N_3$ by (expressions denoting)
numbers in $\naturals$, $R, R_1,$ and $ R_2$ by programs in $\cal L$,
and $M, M', M'',$ and $M'''$ by memory. Second, in keeping with the
systematic confusion of natural numbers with their representation in
$\cal L$, we have also overloaded the operators $+$, $-$, $>$, $\leq$,
and $\in$. Note also that with $-$ we associate the usual subtraction
operation on natural numbers: $N_1 - N_2$ is $0$ if $N_1$ is less than
$N_2$. Finally, these rules make precise the interpretation that we
would naturally think of associating with each of the constructs in
the language. In this regard, the first five rules need no further
explanation. The sixth and seventh rules encode the meaning of memory
lookup and update, respectively: $\deref{R}$ causes $R$ to be
evaluated and the memory to be looked up at the resulting location
while leaving the memory unchanged, whereas $R_1\assign{}R_2$ causes
memory to be changed at the location corresponding to $R_1$ by the
value corresponding to $R_2$.
Notice also that the rightmost premise of the memory update rule
ensures that the domain of memory remains fixed throughout evaluation.
The last three rules make precise the meaning of sequencing and
of $while$ as an iteration construct.

When building derivations, we may build derivations for premises in
any order provided the constrains between premises are met. However,
we may significantly simplify the process of proof construction if we
build them sequentially from the left most premise to the right most
premise. Observe that adopting such a derivation building strategy
does not limit the derivations that can be built.
\subsection{Derivations as Computations}
\label{subsec:derivations}

The rules defining the evaluation semantics provide us a means for
constructing derivations of particular evaluation relations. Such
derivations can be understood as an abstract view of the computation
that results from particular programs. For example, suppose we are given a
particular program $R$ and a starting memory $M$ and we desire to
understand what value this program computes and what impact it has on
memory. In this case, would pick two ``meta variables'' $N$ and $M'$
and we attempt to construct a derivation for the evaluation relation
\[ \evalrelation{R}{M}{N}{M'} \]
with the proviso that we may instantiate $N$ and $M'$ as needed along
the way. Note also that the result of a computation must in fact be
validated by success in constructing such a derivation. Thus, by
analyzing all the possible derivations we also obtain a means for
establishing properties of computations.

To illustrate the connection between derivations and computations in
this setting, let us consider the program
\[\ensuremath{2 \assign \deref{0} ; (0 \assign \deref{1} ; 1 \assign \deref{2})} \]
and its evaluation in some memory
$M$ defined at locations $0,1,$ and $2$.
This program swaps the values stored at two locations using a third
location as temporary storage.
We will build a derivation piecemeal, showing that for some $N$, the
evaluation relation

\begin{tabbing}
\qquad\=\qquad\=\kill
\>$\langle \ensuremath{2 \assign \deref{0} ; (0 \assign \deref{1} ; 1 \assign \deref{2})}, M \rangle \leadsto$\\
\>\>$(N,
\memoryupdate{\memoryupdate{\memoryupdate{M}{2}{M(0)}}{0}{M(1)}}{1}{M(0)})$
\end{tabbing}
  holds.

Let $\Chi$ be the following derivation for
$\evalrelation{\ensuremath{2 \assign \deref{0}}}{M}{M(0)}{M'}$ where
\ensuremath{M'=\memoryupdate{M}{2}{M(0)}}: 
\begin{customprooftree}
          \AxiomC{}
        \UnaryInfC{$\evalrelation{2}{M}{2}{M}$}
          \AxiomC{}
      \UnaryInfC{$\evalrelation{0}{M}{0}{M}$}
    \UnaryInfC{$\evalrelation{\ensuremath{2 \assign \deref{0}}}{M}{M(0)}{M}$}
      \AxiomC{}
    \UnaryInfC{$M(2) = N_1$}
  \TrinaryInfC{$\evalrelation{\ensuremath{2 \assign \deref{0}}}{M}{M(0)}{M'}$}
\end{customprooftree}
Let $\Psi$ be the following derivation for
$\evalrelation{\ensuremath{0 \assign \deref{1}}}{M'}{M'(1)}{M''}$ where
\ensuremath{M''=\memoryupdate{M'}{0}{M'(1)}}: 
\begin{customprooftree}
        \AxiomC{}
      \UnaryInfC{$\evalrelation{0}{M'}{0}{M'}$}
          \AxiomC{}
        \UnaryInfC{$\evalrelation{1}{M'}{1}{M'}$}
      \UnaryInfC{$\evalrelation{\ensuremath{\deref{1}}}{M'}{M'(1)}{M'}$}
        \AxiomC{}
      \UnaryInfC{$M'(0) = N_2$}
    \TrinaryInfC{$\evalrelation{\ensuremath{0 \assign \deref{1}}}{M'}{M'(1)}{M''}$}
\end{customprooftree}
Finally, let $\Omega$ be the following derivation for
$\evalrelation{\ensuremath{1 \assign \deref{2}}}{M''}{M''(2)}{M'''}$ where
\ensuremath{M'''=\memoryupdate{M''}{1}{M''(2)}}: 
\begin{customprooftree}
        \AxiomC{}
      \UnaryInfC{$\evalrelation{1}{M''}{1}{M''}$}
           \AxiomC{}
        \UnaryInfC{$\evalrelation{2}{M''}{2}{M''}$}
      \UnaryInfC{$\evalrelation{\ensuremath{\deref{2}}}{M''}{M''(2)}{M''}$}
        \AxiomC{}
      \UnaryInfC{$M''(1) = N_3$}
    \TrinaryInfC{$\evalrelation{\ensuremath{1 \assign \deref{2}}}{M''}{M''(2)}{M'''}$}
\end{customprooftree}
Then we can combined $\Chi,\Psi$ and $\Omega$ to obtain the complete
derivation that is shown below for the complete program expression of
interest: 
\begin{customprooftree}
  \AxiomC{$\Chi$}
    \AxiomC{$\Psi$}
    \AxiomC{$\Omega$}
  \BinaryInfC{$\evalrelation{\ensuremath{0 \assign \deref{1} ; 1 \assign \deref{2}}}{M'}{M''(2)}{M''}$}
\BinaryInfC{$\evalrelation{\ensuremath{2 \assign \deref{0} ; (0 \assign \deref{1} ; 1 \assign \deref{2})}}{M}{M''(2)}{M'''}$}
\end{customprooftree}
To arrive at the desired conclusion, we have to show that $M'''$, the
memory at the end of the computation, is equivalent to
\[\memoryupdate{\memoryupdate{\memoryupdate{M}{2}{M(0)}}{0}{M(1)}}{1}{M(0)}.\]
Substituting the definition of $M'$ in the definition of $M''$ yields
\[ \memoryupdate{\memoryupdate{M}{2}{M(0)}}{0}{\memoryupdate{M}{2}{M(0)}(1)} .\]
By observing that 
\[ \memoryupdate{M}{2}{M(0)}(1) = M(1) \]
we have
\[ M'' = \memoryupdate{\memoryupdate{M}{2}{M(0)}}{0}{M(1)}.\]
Replacing this result for $M''$ in the definition of $M'''$ we get
\begin{tabbing} \qquad\=\qquad\=\kill
\>$\memoryupdate{\memoryupdate{M}{2}{M(0)}}{0}{M(1)}[1\mapsto$\\ 
\>\>$\memoryupdate{\memoryupdate{M}{2}{M(0)}}{0}{M(1)}(2)]$
\end{tabbing} 
Finally, by observing that 
\[ \memoryupdate{\memoryupdate{M}{2}{M(0)}}{0}{M(1)}(2)= M(0) \] 
we arrive at the conclusion we want:
\[M''' = \memoryupdate{\memoryupdate{\memoryupdate{M}{2}{M(0)}}{0}{M(1)}}{1}{M(0)}.\]

\subsection{Informal Reasoning about Imperative Programs}
\label{subsec:sosreasoning}

As we have explained earlier, we can extract information about the
behavior of a program by analyzing the derivations that result from
it. We illustrate this possibility in this subsection by showing how
to demonstrate the correctness of a program for calculating the sum
of the integers from $0$ to a particular number $N$. Our argument at
this stage will be informal; later sections will discuss a framework
for formalizing this style of argument.

Let $U$ be the following program:
\begin{equation}\label{eqn:sum_loop}
\ensuremath{\soswhile{(\deref{1}) \hat{>} 0}{0 \assign \deref{0} \hat{+} \deref{1} ; 1 \assign \deref{1} \hat{-} 1}}
\end{equation}

Consider the program $V$ written to calculate the value of $\sum\limits_{i=0}^N i$ constructed with $U$:
\begin{equation}\label{eqn:sum}
\ensuremath{0 \assign 0};(\ensuremath{1 \assign N};U)
\end{equation}

We will show that given any $N$ and any memory defined at $0$ and $1$, $V$ calculates the correct answer and stores it in memory.

\begin{lem}[Total Correctness of $U$ using structural operational semantics]\label{lem:sumcorrectnesssoslemma}
$\forall N_1, N_2, M$ if $N_1,N_2\in\naturals$ and $M$ is memory where $M(0) = N_2$ and $M(1) = N_1$ then $\exists M'$ such that $\evalrelation{U}{M}{0}{M'}$ and $M'(0) = N_2+\sum\limits_{i =  0}^N i$
\end{lem}

\begin{proof} [ Proof of Lemma \ref{lem:sumcorrectnesssoslemma} ]

This will be proven by induction on $N_1$.
If $N_1 = 0$ then the following derivation can be constructed and $M'(0) = N_2 = N_2+\sum\limits_{i =  0}^0 i$:
\begin{customprooftree}
        \AxiomC{}
      \UnaryInfC{$\evalrelation{1}{M}{1}{M}$}
    \UnaryInfC{$\evalrelation{\ensuremath{\deref{1}}}{M}{0}{M}$}
      \AxiomC{}
    \UnaryInfC{$\evalrelation{0}{M}{0}{M}$}
      \AxiomC{}
    \UnaryInfC{$0\leq 0$}
  \TrinaryInfC{$\evalrelation{\ensuremath{(\deref{1}) \hat{>} 0}}{M}{0}{M}$}
\UnaryInfC{$\evalrelation{\ensuremath{\soswhile{(\deref{1}) \hat{>} 0}{U}}}{M}{N}{M}$}
\end{customprooftree}

If $M(1) = N_1$ and it is assumed this lemma holds for all memory $W$ where $W(1)<N_1$ then the following derivation can be constructed:
\begin{customprooftree}
    \AxiomC{$\Chi$}
    \AxiomC{}
  \UnaryInfC{$1>0$}
    \AxiomC{$\Psi$}
    \AxiomC{$\Omega$}
  \noLine
  \UnaryInfC{$\evalrelation{\ensuremath{\soswhile{(\deref{1}) \hat{>} 0}{U}}}{M''}{0}{M'''}$}
\QuaternaryInfC{$\evalrelation{\ensuremath{\soswhile{(\deref{1}) \hat{>} 0}{U}}}{M}{0}{M'''}$}
\end{customprooftree}
In this derivation, we let $M'=\memoryupdate{M}{0}{M(0)+M(1)}$, $M''=\memoryupdate{M'}{1}{M(1)-1}$, and $M'''$ is the result memory from our inductive hypothesis.
Let $\Chi$ be a derivation with an end-sequent of $\evalrelation{\ensuremath{(\deref{1}) \hat{>} 0}}{M}{1}{M}$ and $\Psi$ be a derivation with an end-sequent of $\evalrelation{\ensuremath{0 \assign \deref{0} \hat{+} \deref{1} ; 1 \assign \deref{1} \hat{-} 1}}{M}{M(1)-1}{M''}$.
Both of these derivations can be constructed but are omitted; they are uninteresting with respect to this case.
Since $M'(1) < N_1$ the inductive hypothesis can be used to give a derivation for $\Omega$.
\end{proof}

From Lemma \ref{lem:sumcorrectnesssoslemma}, the following theorem is easily shown:
\begin{thm}[Total Correctness of $V$ using structural operational semantics]\label{thm:sumcorrectnesssostheorem}
$\forall N, M$ if $N\in\naturals$ and $M$ is memory defined at $M(0)$ and $M(1)$ then $\exists M'$ such that $\evalrelation{V}{M}{0}{M'}$ and $M'(0) = \sum\limits_{i= 0}^N i$.
\end{thm}
\begin{proof}[ Proof of Theorem \ref{thm:sumcorrectnesssostheorem}]

By case analysis on the derivation for $\evalrelation{V}{M}{0}{M'}$ it suffices to show there is a derivation for $\evalrelation{U}{\memoryupdate{(\memoryupdate{M}{0}{0})}{1}{N}}{0}{M'}$ where $M'(0)=\sum\limits_{i = 0}^N i$.
This is shown using Lemma \ref{lem:sumcorrectnesssoslemma}.

\end{proof}

\section{The Specification Logic}
\label{sec:speclogic}

In this section we present Lolli, the fragment of linear logic that we
will use to formalize our imperative programming language. The first
subsection introduces the language of Lolli and clarifies the meaning
of its logical symbols through inference rules. This part of our
presentation emphasizes the declarative nature of Lolli. When we use
it to model natural semantics style descriptions, we would also like
to be able to capture the structure of natural semantics style
derivations. Towards this end, we show in the second subsection the
relative completeness of goal-directed reasoning in Lolli. This
discussion culminates in a reduced proof system for Lolli that we use
exclusively in the rest of the paper.

\subsection{The logic Lolli}
\label{subsec:lolli}

Lolli is a logic that is built on the simply typed $\lambda$-calculus
of Church~\cite{church40}.
The types underlying its language are constructed from a collection of
primitive types that contain $\omicron$, the type of propositions,
and at least one other type; for the moment, we assume $\iota$ to be
the only such type, but we will add to this collection as needed in
later sections.
The remaining types build on these primitive types using the function
type constructor: if $\tau_1$ and $\tau_2$ are types, then
$\tau_1 \rightarrow \tau_2$ is also a type and it denotes the
collection of functions from $\tau_1$ to $\tau_2$.

The terms of Lolli are constructed from collections of typed variables
and constants using the usual abstraction and application operations:
the former yields the term $\lambda x. t$ of type $\tau_1 \rightarrow
\tau_2$ given the term $t$ of type $\tau_2$ and the variable $x$ of
type $\tau_1$, and the latter yields the term $(t_1\app t_2)$ of type
$\tau_2$ given terms $t_1$ and $t_2$ of types $\tau_1 \rightarrow
\tau_2$ and $\tau_1$ respectively. Abstraction is a binding operation
that defines a scope for the variable, a concept that we will assume
the reader to be familiar with. Two terms are considered to be equal
if one can be obtained from the other by some sequence of
$\alpha$-conversions, i.e. the replacement of a subpart of the form
$\lambda x. t$ by $\lambda y.t'$ provided $x$ and $y$ are variables of
the same type, $y$ does not appear free in $t$ and $t'$ results from
$t$ by the replacement of the free occurrences of $x$ by $y$. Given a
term $s$ of the same type as $x$, we will write $t[s/x]$ to denote the
result of substituting $s$ for the free occurrences of $x$ in $t$ in a
capture avoiding way; notice that in correctly carrying out such a
substitution, we may need to apply some $\alpha$-conversions. A term
$t$ is said to be obtained by $\beta$-contraction from another term
$s$ if it results from replacing a subterm of $s$ that has the form
$((\lambda x.t_1)\app t_2)$ by $t_2[t_1/x]$. Two terms are also
considered equal if one can be obtained from the other by some
sequence of applications of $\beta$-contractions or its inverse. We
will use this notion of equality implicitly in the rest of this
paper. In the context of the simply typed $\lambda$-calculus, it is
known that every term has a normal form modulo $\beta$-contractions,
i.e. it is equal to a term which does not contain a subterm of the
form $((\lambda x.t_1)\app t_2)$. We will depict terms solely by their
normal forms.

Lolli has a set of constants that serve to build a logic over its
terms. These constants, referred to as {\em logical constants} consist
of the following: $\AMPCON$, $\LI$, $\I$, $\OT$, and $\OP$ all of type
$\omicron \rightarrow (\omicron \rightarrow \omicron)$ and written
in infix form; $!$ of type $\omicron \rightarrow \omicron$;
and, for each type $\tau$, the constants $\A_\tau$ and $\E_\tau$ with
type $(\tau\rightarrow\omicron)\rightarrow\omicron$.
The constants $\A_\tau$ and $\E_\tau$ are referred to as quantifiers
and the remaining constants constitute the logical connectives.
In addition to these constants, expressions in Lolli may also be
formed from user defined constants, referred to as {\em nonlogical
  constants}. The well-formed terms of type $\omicron$ in Lolli are
distinguished as {\em formulas}. Notice that a formula may have as its
top-level symbol a logical constant, a variable or a nonlogical
constant. In the latter two cases, the formula is said to be {\em
  atomic}. Further, it is a rigid atom if its top-level symbol is a
nonlogical constant. We shall use the syntactic variable $A$ to denote
atomic formulas and $A_r$ to denote rigid atoms.

At a logical level, Lolli is oriented towards proving judgments
represented by {\em sequents}. Formally, a sequent is an object of the
form
\[\sequent{\Gamma}{\Delta}{G} \]
where $\Gamma$ is a set of formulas, $\Delta$ is a multiset of
formulas and $G$ is a formula. Intuitively, such a sequent corresponds
to the claim that $G$, the {\it goal formula}, is derivable given the
resources $\Gamma$ and $\Delta$. The resources in $\Gamma$ are
distinguished as being {\it unbounded}: formulas in $\Gamma$ would
typically be used to represent unchanging facts in a specification
setting, such as the natural semantics rules governing the behavior of
imperative programs. On the other hand, formulas in $\Delta$
constitute {\em bounded} resources: referring again to the imperative
programming example, they may be used to represent the state of memory
at a particular point in computation.

The syntax of formulas that may be used as resources and goals is
limited in Lolli. Specifically, they may only be the $P$ and $G$
formulas described by the syntax rules below:
\begin{equation}\label{eqn:lollisyntax}
\begin{split}
P::=&\;\; A_r \;|\; P \AMPOP P \;|\; G \LI P \;|\; G \I P \;|\; \A x . P\\
G::=&\;\; \top \;|\; A \;|\; G \AMPOP G \;|\; P \LI G \;|\; P \I G \;|\; \A x . G\\
&|\;\E x. G \;|\; !G \;|\; G \OT G \;|\; G \OP G
\end{split}
\end{equation}
We refer to $P$ formulas also as \emph{program clause formulas}.
Notice that the connectives $\AMPCON$, $\LI$, $\I$, and $\A$ are allowed in
both kinds of formulas. However, there are differing constraints in
the use of $\LI$ and $\I$. When these are used in the resource
formulas, the formula on the left must be a goal formula and that on
the right must be a resource formula. When they are used in a goal
formula on the other hand, the formula on the left must be a resource
formula and that on the right must be a goal formula. As we shall see
presently, these restrictions play an important role in maintaining
the structure of sequents in the course of a derivation and therefore
in the coherence of the inference rules for Lolli. In addition to the
already mentioned connectives, goal formulas may contain $\top$, $\E$,
$!$, $\OT$, and $\OP$.

The rules for deriving sequents in Lolli are presented in Figure
\ref{fig:lollirules}.
The sequent that appears below the line in each of these rules is
called its conclusion and the sequents that appear above the line
constitute its premises.
The $L$ or $R$ in the labels of these rules denotes whether the rule
introduces a logical symbol on the left or the right of the $\turnstile$.
Grouped by $L$ or $R$ they may be referred to as left-introduction
rules and right-introduction rules, respectively.
In the rules pertaining to the logical symbols, the formula in the
conclusion that contains the introduced symbol is called the {\it
  principal formula}. This terminology is extended to the $id$ rule
and $absorb$ rules to denote the formulas represented by $A$ and $B$,
respectively. When we write $\Gamma, F$ in the unbounded context in
these rules, we mean it to denote $\Gamma \cup \{F\}$, i.e. $F$ may
also be contained in $\Gamma$. On the other hand, in the unbounded context
$\Delta, F$ represents $\Delta \uplus \{F\}$, i.e. $\Delta$
constitutes the bounded resources with the exclusion of the selected
copy of the formula $F$. Relatedly, $\Delta_1, \Delta_2$ in such a
context stands for $\Delta_1 \uplus \Delta_2$, i.e., the comma
represents multiset union.

\begin{figure*}
\begin{framed}
\begin{customprooftree}

  \AxiomC{}
  \RightLabel{$id$}
  \UnaryInfC{$\Gamma ; A \SQ A$}

  \AxiomC{$\Gamma,B; \Delta,B \SQ G$}
  \RightLabel{$absorb$}
  \UnaryInfC{$\Gamma,B; \Delta \SQ G$}

  \AxiomC{}
  \RightLabel{$\top R$}
  \UnaryInfC{$\Gamma ; \Delta \SQ \top$}

  \AxiomC{$\Gamma;\Delta,B_i\SQ G$}
  \RightLabel{$\&L(i\in\{1,2\})$}
  \UnaryInfC{$\Gamma; \Delta,B_1\& B_2\SQ G$}

  \AxiomC{$\Gamma; \Delta\SQ G_1$}
  \AxiomC{$\Gamma; \Delta\SQ G_2$}
  \RightLabel{$\&R$}
  \BinaryInfC{$\Gamma; \Delta \SQ G_1\&G_2$}

\alwaysNoLine
\QuinaryInfC{}

\end{customprooftree}

\begin{customprooftree}

  \AxiomC{$\Gamma; \Delta_1\SQ B_1$}
  \AxiomC{$\Gamma; \Delta_2,B_2\SQ G$}
  \RightLabel{$\LI L$}
  \BinaryInfC{$\Gamma;\Delta_1,\Delta_2,B_1\LI B_2\SQ G$}

  \AxiomC{$\Gamma; \Delta,G_1\SQ G_2$}
  \RightLabel{$\LI R$}
  \UnaryInfC{$\Gamma; \Delta \SQ G_1\LI G_2$}

  \AxiomC{$\Gamma; \emptyset \SQ B_1$}
  \AxiomC{$\Gamma; \Delta,B_2\SQ G$}
  \RightLabel{$\I L$}
  \BinaryInfC{$\Gamma;\Delta,B_1\I B_2\SQ G$}

  \AxiomC{$\Gamma,G_1; \Delta \SQ G_2$}
  \RightLabel{$\I R$}
  \UnaryInfC{$\Gamma; \Delta \SQ G_1\I G_2$}

\alwaysNoLine
\QuaternaryInfC{}

\end{customprooftree}

\begin{customprooftree}

  \AxiomC{$\Gamma;\Delta, (B\app t)\SQ G$}
  \RightLabel{$\A L$}
  \UnaryInfC{$\Gamma;\Delta,\A x . B\SQ G$}

  \AxiomC{$\Gamma;\Delta\SQ G c$}
  \RightLabel{$\A R$}
  \UnaryInfC{$\Gamma;\Delta\SQ \A x . G$}

  \AxiomC{$\Gamma;\Delta \SQ (G\app t)$}
  \RightLabel{$\E R$}
  \UnaryInfC{$\Gamma;\Delta \SQ \E x . G$}

\alwaysNoLine
\TrinaryInfC{}

\end{customprooftree}

\begin{customprooftree}

  \AxiomC{$\Gamma; \emptyset\SQ G$}
  \RightLabel{$!R$}
  \UnaryInfC{$\Gamma; \emptyset \SQ !G$}

  \AxiomC{$\Gamma; \Delta \SQ G_i$}
  \RightLabel{$\OP R(i\in\{1,2\})$}
  \UnaryInfC{$\Gamma; \Delta \SQ G_1\OP G_2 $}

  \AxiomC{$\Gamma; \Delta_1 \SQ G_1$}
  \AxiomC{$\Gamma; \Delta_2 \SQ G_2$}
  \RightLabel{$\OT R$}
  \BinaryInfC{$\Gamma; \Delta_1,\Delta_2 \SQ G_1\OT G_2$}

\alwaysNoLine
\TrinaryInfC{}

\end{customprooftree}

\caption{
The inference rules in Lolli. In the $\A R$ rule, $c$ must not occur
in $\Gamma$, $\Delta$, or $G$. In the $\A L$ and $\E R$
rules, the term $t$
generalized upon must be such that $(B\app t)$ and $(G\app t)$ are a program
clause formula and a goal formula, respectively.
}
\label{fig:lollirules}
\end{framed}
\end{figure*}

Some comments on the inference rules are useful both in understanding
the logical structure of Lolli and the intended meaning of the logical
symbols.
The rules for the use of resource formulas in Lolli are all stated
with respect to the bounded context.
The only exception to this is the $absorb$ rule which encodes the
possibility of making a copy of an unbounded resource before using it
in a bounded fashion.
The rules for the quantifiers give them their usual interpretation
with the caveat that the domain of quantification is restricted so as
to preserve the normal form of sequents in Lolli.
The formula $G_1 \OT G_2$ is interpreted as saying that there
are enough resources to show both $G_1$ and $G_2$: the rule for
proving this formula requires each component to be shown from a
partitioning of the bounded resources.
The connectives $\AMPCON$ and $\OP$ are meant to encode different
kinds of choices. 
The formula $G_1 \AMPOP G_2$ signifies that the available
resources are sufficient to satisfy either $G_1$ or $G_2$, whichever
one we choose. 
Accordingly, to prove a sequent that has such a formula on the right
of $\vdash$, we have to show that we can prove sequents with the same
resources and each of $G_1$ and $G_2$ as the goal. 
On the other hand, if the formula $B_1 \AMPOP B_2$ is available as a
resource, this means that we can choose which one of the components we
actually want to use, something that underlies the left-introduction
rule for this connective. 
In contrast, the formula $G_1 \OP G_2$ means that we can have one of
$G_1$ or $G_2$ based on the resources, but we do not know {\em a
  priori} which. 
Correspondingly, to prove a sequent that has such a formula on the right
of $\vdash$, it suffices to prove a sequent with the same
resources and with one of $G_1$ or $G_2$ as the goal. 
The $\LI$ connective captures a notion of resource conversion: To show
$G_1\LI G_2$  we must somehow use $G_1$ in showing $G_2$ and,
conversely, when given $B_1 \LI B_2$, we may consume some of the
resources to show $B_1$ and then use $B_2$ itself as a resource.
The $\I$ connective also represents resource conversion, but this time
an unbounded resource.
Note that the rules for $\LI$ and $\I$ may move formulas from one side
of $\turnstile$ to the other and could potentially result in
destroying the form of permitted sequents in Lolli.
However, the restriction on what can appear on either side of $\LI$
and $\I$ in goal and program clause formulas ensures that this does
not happen.
The $!$ connective corresponds to treating its argument as being
independent of the finite resources.
The $id$ rule cements the fact that all the bounded resources must be
consumed in a derivation.
In this setting $\top$ corresponds to a ``sink'' or a garbage
collector for the bounded resources.

We illustrate the rules of Lolli by considering a few proofs that use
them. First, consider the sequent
\[\sequent{\emptyset}{\emptyset}{(A_1 \AMPOP A_2) \I (A_1 \OT A_2)}. \]
This sequent expresses the intuition that if we
have $A_1 \AMPOP A_2$ as an unbounded resource, then we must
simultaneously have both $A_1$ and $A_2$ provided our bounded
resources are empty. A derivation for the sequent is shown below.
\begin{customprooftree}\label{prooftree:nup}

    \AxiomC{}
    \RightLabel{$id$}
    \UnaryInfC{$A_1 \AMPOP A_2 ; A_2 \SQ A_2 $}

    \AxiomC{}
    \RightLabel{$id$}
    \UnaryInfC{$A_1 \AMPOP A_2 ; A_1 \SQ A_1 $}

  \RightLabel{$\OT R$}
  \BinaryInfC{$A_1 \AMPOP A_2 ; A_1, A_2 \SQ A_1 \OT A_2 $}
  \RightLabel{$\AMPCON L$}
  \UnaryInfC{$A_1 \AMPOP A_2 ; A_1, A_1 \AMPOP A_2 \SQ A_1 \OT A_2 $}
  \RightLabel{$\AMPCON L$}
  \UnaryInfC{$A_1 \AMPOP A_2 ; A_1 \AMPOP A_2, A_1 \AMPOP A_2 \SQ A_1 \OT A_2 $}
  \RightLabel{$absorb$}
  \UnaryInfC{$A_1 \AMPOP A_2 ; A_1 \AMPOP A_2 \SQ A_1 \OT A_2 $}
  \RightLabel{$absorb$}
  \UnaryInfC{$A_1 \AMPOP A_2 ; \emptyset \SQ A_1 \OT A_2 $}
  \RightLabel{$\I R$}
  \UnaryInfC{$\emptyset ; \emptyset \SQ (A_1 \AMPOP A_2) \I (A_1 \OT A_2)$}

\end{customprooftree}
This proof uses the $\AMPCON L$ and $absorb$ rules in a situation when the
formula on the right $\turnstile$ is $A_1\OT A_2$, i.e., is not atomic.
Such a proof is not goal-directed, i.e., if we think of the process of
searching for a proof for the given sequent, the formula on the right
of the $\turnstile$ symbol does not guide the choice of rule to use to
arrive at the conclusion. In the next subsection we will consider the
idea of uniform proofs that will provide us a means for restricting
attention to only goal-directed proofs.

Notice that the unbounded availability of $A_1 \AMPOP A_2$ is important to
the above proof: if we change the sequent to
\[\sequent{\emptyset}{\emptyset}{(A_1 \AMPOP A_2) \LI (A_1 \OT A_2)} \]
then it is no longer provable.
Given the formula $A_1 \AMPOP A_2$ as a bounded resource, we have to make
a choice between using with $A_1$ or $A_2$, also as a bounded resource.
This choice is mutually exclusive; we may not have both $A_1$ and
$A_2$. On the other hand, if $A_1 \OT A_2$ is on the right of
$\turnstile$, then both $A_1$ and $A_2$ must be available as (bounded)
resources for the sequent to be provable.

The process of finding proofs for sequents typically involves search.
Two common strategies that are used in this setting are \emph{forward
  chaining} and \emph{backward chaining}.
These strategies refer to how we use implicational formulas, which in Lolli
could be ones that have either $\LI$ or $\I$ as their top-level
connective, available in resources in guiding the search. In the
former case, we use the fact that the lefthand side of the implication
is already available as a resource and we then reason forward, by
adding the righthand side as a resource. In the latter case, we
observe that the goal formula of the sequent matches the righthand
side of the implication and then reduce the task to showing the
lefthand side from the available resources.
The following proof can be understood as the result of using a forward
chaining strategy to prove the sequent $\sequent{A_1}{A_1\LI A_2,
  A_2\LI A_3}{A_3}$.
\begin{customprooftree}

  \AxiomC{}
  \RightLabel{$id$}
  \UnaryInfC{$A_1;{A_1}\SQ A_1$}
  \RightLabel{$absorb$}
  \UnaryInfC{$A_1;\emptyset \SQ A_1$}

   \AxiomC{}
   \RightLabel{$id$}
   \UnaryInfC{$A_1;{A_2}\SQ A_2$}

   \AxiomC{}
   \RightLabel{$id$}
   \UnaryInfC{$A_1;{A_3}\SQ A_3$}

  \RightLabel{$\LI L$}
  \BinaryInfC{$A_1;A_2,{A_2\LI A_3}\SQ A_3$}

  \RightLabel{$\LI L$}
  \BinaryInfC{$A_1; {A_1\LI A_2}, A_2\LI A_3\SQ A_3$}

\end{customprooftree}

In Section \ref{subsec:speclogic}, we will consider a different proof
resulting from a backward chaining strategy for this sequent.

\subsection{A Reduced Proof System for Lolli}
\label{subsec:speclogic}

The derivation system for Lolli that we saw in the previous subsection
presents us with alternative ways to construct a proof. For example,
we may have the choice of using either a left or a right rule at a
particular point in proof. In modelling natural semantics style rules
for imperative programming languages, we will want to use sequents in
a specific way: the unbounded context will encode the semantics of
programming constructs, the bounded context will model the state and
the goal formula will represent the program producing the
computation. If we are to analyze the properties of programs using
this setup, it would be ideal if we could focus our attention on Lolli
proofs that closely follow program behavior. We show here that this is
possible. In particular, we demonstrate that, from a provability
perspective, it suffices to look at proofs that are goal-directed in
that, when looking at derivations bottom up, the first step is always
to simplify a complex goal formula.

The following definition, first introduced by Miller\emph{et al}\cite{miller91apal}, provides an encapsulation of the idea of
goal-directedness in the context of Lolli proofs. It was or
\begin{defn}[Uniform Proof]\label{def:uniformprovability} A
  uniform proof is a Lolli proof in which every sequent with a non-atomic
  goal formula on the right of $\turnstile$ is the conclusion of an
  inference rule that introduces   the top-level logical symbol of
  that formula.
\end{defn}
Towards understanding uniform provability, consider the proof shown in
Section \ref{subsec:lolli} for the sequent
\[\sequent{\emptyset}{\emptyset}{(A_1 \AMPOP A_2) \I (A_1 \OT A_2)}. \]
That proof is not a uniform proof.
In that proof, there are two $absorb$ rules and two $\AMPCON L$ rules that
have as a conclusion a sequent in which the goal formula $A_1 \OT A_2$
appears as the right of $\turnstile$. However, the same sequent does
have a uniform proof that is shown below:
\begin{customprooftree}

    \AxiomC{}
    \RightLabel{$id$}

    \UnaryInfC{$A_1 \AMPOP A_2 ; A_1 \SQ A_1 $}
    \RightLabel{$\AMPCON L$}
    \UnaryInfC{$A_1 \AMPOP A_2 ; A_1 \AMPOP A_2 \SQ A_1 $}
    \RightLabel{$absorb$}
    \UnaryInfC{$A_1 \AMPOP A_2 ; \emptyset \SQ A_1 $}

    \AxiomC{}
    \RightLabel{$id$}
    \UnaryInfC{$A_1 \AMPOP A_2 ; A_2 \SQ A_2 $}
    \RightLabel{$\AMPCON L$}
    \UnaryInfC{$A_1 \AMPOP A_2 ; A_1 \AMPOP A_2 \SQ A_2 $}
    \RightLabel{$absorb$}
    \UnaryInfC{$A_1 \AMPOP A_2 ; \emptyset \SQ A_2 $}

  \RightLabel{$\OT R$}
  \BinaryInfC{$A_1 \AMPOP A_2 ; \emptyset \SQ A_1 \OT A_2$}
  \RightLabel{$\I R$}
  \UnaryInfC{$\emptyset ; \emptyset \SQ (A_1 \AMPOP A_2) \I (A_1 \OT A_2)$}

\end{customprooftree}
In fact, every provable Lolli sequent has a uniform proof as we now show.

\begin{thm}[Lolli Admits Uniform Provability]\label{thm:uniformprovability}
The sequent $\ubsequent{G}$ has a proof in Lolli if and only if it has a uniform proof.
\end{thm}

\begin{proof} 
The ``if'' direction is obvious. For the ``only if'' direction, we
consider a proof that is not uniform and show how to transform it into
a uniform proof. We associate with a proof a non-uniformity measure
that counts the number of inference rule occurrences that do not act
on a complex goal formula that appears to the right of $\turnstile$ in
their conclusion. If this measure is non-zero, we show how to reduce
it by 1. The conclusion then follows by induction on the measure.

If a proof has a non-zero non-uniformity measure, then there must be a
path in it in which there is a first occurrence of a left rule that
has a complex goal formula to the right of $\turnstile$ in its
conclusion. We show how to reduce the height of this path by 1. By
induction on this height it follows that we can eliminate this
violation of uniformity and thereby reduce the non-uniformity measure
of the proof. Observe that since the rule in question is the first one
along the path to violate the uniformity property, it must be preceded
in the proof by a right rule. We use this fact in our argument. In
particular, we consider the possible cases for the right and left
rules and show that the left rule can be permuted above the right one,
thereby moving the violation of non-uniformity closer to a leaf.

In a detailed consideration of the cases, it is useful to categorize
rules based on the number of premises they have. Category \rom{1} will
represent rules with one premise and category \rom{2} will represent
rules with two premises.

Suppose that the case in question involves two inference rules
from category \rom{1}. An example of such a situation is the following:
\begin{center}
\begin{customprooftree}\label{prooftree:nup1}

  \AxiomC{$\Xi$}
  \alwaysNoLine
  \UnaryInfC{$\Gamma;\Delta,B_i, G_1 \SQ G_2$}
  \alwaysSingleLine
  \RightLabel{$\LI R$}
  \UnaryInfC{$\Gamma;\Delta,B_i\SQ G_1\LI G_2$}
  \RightLabel{$\AMPCON L(i\in\{1,2\})$}
  \UnaryInfC{$\Gamma;\Delta,B_1\AMPOP B_2 \SQ G_1\LI G_2$}

\end{customprooftree}
\end{center}
This proof can be rearranged as follows:
\begin{center}
\begin{customprooftree}

  \AxiomC{$\Xi$}
  \alwaysNoLine
  \UnaryInfC{$\Gamma;\Delta,B_i, G_1 \SQ G_2$}
  \alwaysSingleLine
  \RightLabel{$\AMPCON L(i\in\{1,2\})$}
  \UnaryInfC{$\Gamma;\Delta,B_1 \AMPOP B_2, G_1 \SQ G_2$}
  \RightLabel{$\LI R$}
  \UnaryInfC{$\Gamma;\Delta,B_1 \AMPOP B_2 \SQ G_1 \LI G_2$}

\end{customprooftree}
\end{center}
By permuting the left rule above the right one, we have reduced the
length of the path by 1 as required.
A similar argument applies to all
the other cases of rules in these two respective categories.

Suppose that the case in question involves a right inference rule from
category \rom{1} and a left inference rule from category \rom{2}. An
example of this kind is presented by the following derivation:
\begin{center}
\begin{customprooftree}

    \AxiomC{$\Psi$}
    \alwaysNoLine
    \UnaryInfC{$\Gamma;\Delta_1 \SQ B_1$}
    \alwaysSingleLine

    \AxiomC{$\Xi$}
    \alwaysNoLine
    \UnaryInfC{$\Gamma;\Delta_2,B_2\SQ G_i$}
    \alwaysSingleLine
    \RightLabel{$\OP R(i\in\{1,2\})$}
    \UnaryInfC{$\Gamma;\Delta_2,B_2\SQ G_1\OP G_2$}

  \RightLabel{$\LI L$}
  \BinaryInfC{$\Gamma;\Delta_1,\Delta_2,B_1\LI B_2\SQ G_1 \OP G_2$}

\end{customprooftree}
\end{center}

In this case the derivation can be rearranged as follows to once again
reduce the length of the path by 1:
\begin{center}
\begin{customprooftree}

    \AxiomC{$\Psi$}
    \alwaysNoLine
    \UnaryInfC{$\Gamma;\Delta_1 \SQ B_1$}
    \alwaysSingleLine

    \AxiomC{$\Xi$}
    \alwaysNoLine
    \UnaryInfC{$\Gamma;\Delta_2,B_2 \SQ G_i$}
    \alwaysSingleLine

  \RightLabel{$\LI L$}
  \BinaryInfC{$\Gamma;\Delta_1,\Delta_2,B_1\LI B_2,\SQ G_i$}

  \RightLabel{$\OP R(i\in\{1,2\})$}
  \UnaryInfC{$\Gamma;\Delta_1,\Delta_2,B_1\LI B_2\SQ G_1 \OP G_2$}

\end{customprooftree}
\end{center}
The other cases for rules from the categories under consideration are similar.

Suppose the case in question involves a right inference rule from category
\rom{2} and a left inference rule instance from category \rom{1}. An
example of this kind is the following:
\begin{center}
\begin{customprooftree}\label{prooftree:nup3}

    \AxiomC{$\Psi$}
    \alwaysNoLine
    \UnaryInfC{$\Gamma;\Delta_1,B_i \SQ G_1$}
    \alwaysSingleLine

    \AxiomC{$\Xi$}
    \alwaysNoLine
    \UnaryInfC{$\Gamma;\Delta_2 \SQ G_2$}
    \alwaysSingleLine

  \RightLabel{$\OT R$}
  \BinaryInfC{$\Gamma;\Delta_1,\Delta_2,B_i\SQ G_1 \OT G_2$}

  \RightLabel{$\AMPCON L (i\in\{1,2\})$}
  \UnaryInfC{$\Gamma;\Delta_1,\Delta_2,B_1\AMPOP B_2\SQ G_1\OT G_2$}

\end{customprooftree}
\end{center}
Here again, we can permute the left inference rule above the right one
as follows:
\begin{center}
\begin{customprooftree}

    \AxiomC{$\Psi$}
    \alwaysNoLine
    \UnaryInfC{$\Gamma;\Delta_1,B_i \SQ G_1$}
    \alwaysSingleLine
    \RightLabel{$\AMPCON L$}
    \UnaryInfC{$\Gamma;\Delta_1,\Delta_2,B_1\AMPOP B_2\SQ G_1\OT G_2$}

    \AxiomC{$\Xi$}
    \alwaysNoLine
    \UnaryInfC{$\Gamma;\Delta_2 \SQ G_2$}
    \alwaysSingleLine

  \RightLabel{$\OT R$}
  \BinaryInfC{$\Gamma;\Delta_1,\Delta_2,B_1\AMPOP B_2\SQ G_1 \OT G_2$}

\end{customprooftree}
\end{center}
The other cases under this combination are treated similarly.

Finally, suppose that the situation under consideration involves a
right and left inference rule both from category \rom{2}. An example
of this kind is the following:
\begin{center}
\begin{customprooftree}

   \AxiomC{$\Psi$}
   \alwaysNoLine
   \UnaryInfC{$\Gamma ; \Delta_1 \SQ B_1$}

    \AxiomC{$\Xi$}
    \alwaysNoLine
    \UnaryInfC{$\Gamma ;\Delta_2, B_2 \SQ G_1$}
    \alwaysSingleLine

    \AxiomC{$\Theta$}
    \alwaysNoLine
    \UnaryInfC{$\Gamma ;\Delta_3 \SQ G_2$}
    \alwaysSingleLine

   \RightLabel{$\OT R$}
   \BinaryInfC{$\Gamma ;\Delta_2, \Delta_3, B_2 \SQ G_1 \OT G_2$}

  \RightLabel{$\LI L$}
  \BinaryInfC{$\Gamma ; \Delta_1,\Delta_2,\Delta_3, B_1 \LI B_2 \SQ G_1 \OT G_2$}

\end{customprooftree}
\end{center}
Here we rearrange the derivation as follows, again obviously reducing
the length of the path to the errant left rule by one.
\begin{center}
\begin{customprooftree}

    \AxiomC{$\Psi$}
    \alwaysNoLine
    \UnaryInfC{$\Gamma ; \Delta_1 \SQ B_1$}

    \AxiomC{$\Xi$}
    \alwaysNoLine
    \UnaryInfC{$\Gamma ;\Delta_2, B_2 \SQ G_1$}
    \alwaysSingleLine

   \RightLabel{$\LI L$}
   \BinaryInfC{$\Gamma ;\Delta_1, \Delta_2, B_1 \LI B_2 \SQ G_1$}

   \AxiomC{$\Theta$}
   \alwaysNoLine
   \UnaryInfC{$\Gamma ;\Delta_3 \SQ G_2$}
   \alwaysSingleLine

  \RightLabel{$\OT L$}
  \BinaryInfC{$\Gamma ; \Delta_1,\Delta_2,\Delta_3, B_1 \LI B_2 \SQ G_1 \OT G_2$}

\end{customprooftree}
\end{center}
The other cases for the rules in the category under consideration are
treated similarly.
\end{proof}

We are thinking of modeling natural semantics style inference rules
using the $\LI$ connective: modeled natural semantics rule conclusion
relations will occur to the right, also known as the {\em head}, of a
$\LI$ and modeled premise relations will occur to left, also known as
the {\em body}, of a $\LI$. When modeled this way, $\LI L$ application
on formulas with heads matching atomic goals mimics natural semantics
style derivation construction. A backward chaining proof search
strategy is one where this process is repeated for proof construction.

The following definition captures the structure of proofs built using
a backward chaining proof search strategy.
\begin{defn}[Simple Proof]\label{def:simpleproofs}
A uniform proof is simple if every left introduction inference rule
instance acts on a marked formula. A unique formula in the bounded
context is marked if it is the principal formula of an $id$ instance or if:
\begin{itemize}
\item $P_1$ or $P_2$ are marked in the premises sequent of a $\AMPCON L$
instance then the formula $P_1\AMPOP P_2$ is marked in the conclusion
sequent.
\item $P[t/x]$ is marked in the premise sequent of a $\A L$ instance
then the formula $\A x.P$ is marked in the conclusion sequent.
\item $P$ is marked in the right-hand premise sequent of a $\LI L$
instance then formula $G \LI P$ is marked in the conclusion sequent.
\item $P$ is marked in the right-hand premise sequent of a $\I L$
instance then formula $G \I P$ is marked in the conclusion sequent.
\end{itemize}
\end{defn}

The second proof from Section \ref{subsec:lolli} is an example of a
proof that is not simple. This can be illustrated by attempting to
mark the proof according Definition \ref{def:simpleproofs}. In the
following proof, the dots indicate formulas which can be marked.
\begin{customprooftree}

  \AxiomC{}
  \RightLabel{$id$}
  \UnaryInfC{$A_1;\dot{A_1}\SQ A_1$}
  \RightLabel{$absorb$}
  \UnaryInfC{$A_1;\emptyset \SQ A_1$}

   \AxiomC{}
   \RightLabel{$id$}
   \UnaryInfC{$A_1;\dot{A_2}\SQ A_2$}

   \AxiomC{}
   \RightLabel{$id$}
   \UnaryInfC{$A_1;\dot{A_3}\SQ A_3$}

  \RightLabel{$\LI L$}
  \BinaryInfC{$A_1;A_2,\dot{A_2\LI A_3}\SQ A_3$}

  \RightLabel{$\LI L$}
  \BinaryInfC{$A_1; A_1\LI A_2,A_2\LI A_3\SQ A_3$}

\end{customprooftree}

Consider the bottom most $\LI L$ instance principal formula $A_1\LI
A_2$, call this instance one.  According to the marking strategy, for
this formula to be marked $A_2$ must be marked at the root of the
right-hand premise sub-proof of instance one.  Consequently, instance
one acts on a unmarked formula.

The following proof is a simple proof for the same sequent. Observe
that every principal formula of a left introduction rule is marked.
\begin{customprooftree}

    \AxiomC{}
    \RightLabel{$id$}
    \UnaryInfC{$A_1;\dot{A_1}\SQ A_1$}
    \RightLabel{$Absorb$}
    \UnaryInfC{$A_1;\emptyset \SQ A_1$}

    \AxiomC{}
    \RightLabel{$id$}
    \UnaryInfC{$A_1;\dot{A_2}\SQ A_2$}

   \RightLabel{$\LI L$}
   \BinaryInfC{$A_1;\dot{A_1\LI A_2}\SQ A_2$}

   \AxiomC{}
   \RightLabel{$id$}
   \UnaryInfC{$A_1;\dot{A_3}\SQ A_3$}

  \RightLabel{$\LI L$}
  \BinaryInfC{$A_1; A_1\LI A_2, \dot{A_2\LI A_3}\SQ A_3$}
\end{customprooftree}

We now show that every provable sequent in Lolli has a simple proof.
\begin{thm}[The Original Specification Logic Admits
Simple Provability]\label{thm:simpleprovability} 
The sequent $\ubsequent{G}$ has a uniform proof in Lolli if and only
if it has a simple proof in Lolli.
\end{thm}

\begin{proof}
This proof is similar to the proof given in Theorem
\ref{thm:uniformprovability}. The ``if'' direction is obvious and in
the ``only if'' direction, we associate with a proof a non-simple
measure that counts the number of unmarked principal formulas
occurring to the left of a $\turnstile$. If this measure is non-zero,
we show how to reduce it by 1. The conclusion then follows by
induction on the measure.

Observe that if a non-simple instance occurs below an $absorb$
instance a permutation is immediate. Therefore, we restrict analysis
to non-simple instances below instances that are not $absorb$.

Suppose the non-simple instance is a $\AMPCON$ or $\A$ left introduction
instance. Observe, that the rule above this non-simple instance must
be a left introduction instance; $A$ is atomic so no right introduction rules
apply. If the rule is below a $\A$ or $\AMPCON$ left introduction instance
permutation of these instances is immediate. If the rule above is a
left introduction instance of $\LI$ a straightforward permutation is
possible. We consider one such case in detail where $\Psi$ and $\Xi$
are simple proofs.
\begin{center}
\begin{customprooftree}

   \AxiomC{$\Psi$}
   \alwaysNoLine
   \UnaryInfC{$\Gamma ; \Delta_1,P_i \SQ B_1$}
   \alwaysSingleLine

   \AxiomC{$\Xi$}
   \alwaysNoLine
   \UnaryInfC{$\Gamma ; \Delta_2, P_2 \SQ A$}
   \alwaysSingleLine

  \RightLabel{$\LI L$}
  \BinaryInfC{$\Gamma ;\Delta_1,\Delta_2,P_1\LI P_2, P_i\SQ A$}
  \RightLabel{$\AMPCON L(i\in\{3,4\})$}
  \UnaryInfC{$\Gamma ; \Delta_1,\Delta_2, P_1 \LI P_2, P_3 \AMPOP P_4 \SQ A$}

\end{customprooftree}
\end{center}
This non-simple uniform proof may be permuted to one with the following
form:
\begin{center}
\begin{customprooftree}

   \AxiomC{$\Psi$}
   \alwaysNoLine
   \UnaryInfC{$\Gamma ; \Delta_1, P_i \SQ P_1$}
   \alwaysSingleLine
   \RightLabel{$\AMPCON L$}
   \UnaryInfC{$\Gamma; \Delta_1, P_3\AMPOP P_4 \SQ P_1$}

   \AxiomC{$\Xi$}
   \alwaysNoLine
   \UnaryInfC{$\Gamma ; \Delta_2, P_2 \SQ A$}
   \alwaysSingleLine

  \RightLabel{$\LI L$}
  \BinaryInfC{$\Gamma ; \Delta_1, \Delta_2, P_1 \LI P_2, P_3 \AMPOP P_4 \SQ A$}

\end{customprooftree}
\end{center}
When the non-simple instance is a $\A$ left introduction instance or
the instance above is a $\I$ the permutations differ only slightly.

Suppose the non-simple instance is a $\LI$ or $\I$ left introduction
instance. Observe, that the right premise must begin with a left
introduction instance; $A$ is atomic so no right introduction rules
apply. Furthermore, observe that the left premise is irrelevant with
respect to marking. We consider one case in detail where $\Psi$
and $\Xi$ are simple proofs.
\begin{center}
\begin{customprooftree}

    \AxiomC{$\Chi$}
    \noLine{}
    \UnaryInfC{$\Gamma ; \Delta_3 \SQ P_3$}
    
     \AxiomC{$\Psi$}
     \noLine
     \UnaryInfC{$\Gamma ; \Delta_1,P_4 \SQ P_1$}

     \AxiomC{$\Xi$}
     \noLine
     \UnaryInfC{$\Gamma ; \Delta_2, P_2 \SQ A$}

    \RightLabel{$\LI L$}
    \BinaryInfC{$\Gamma ;\Delta_1,\Delta_2,P_1\LI P_2, P_4\SQ A$}

  \RightLabel{$\LI L$}
  \BinaryInfC{$\Gamma ; \Delta_1,\Delta_2, \Delta_3, P_1 \LI P_2, P_3 \LI P_4 \SQ A$}

\end{customprooftree}
\end{center}
This non-simple uniform proof may be permuted to one with the
following form:
\begin{center}
\begin{customprooftree}

     \AxiomC{$\Chi$}
     \noLine{}
     \UnaryInfC{$\Gamma ; \Delta_3 \SQ P_3$}
    
     \AxiomC{$\Psi$}
     \noLine
     \UnaryInfC{$\Gamma ; \Delta_1,P_4 \SQ P_1$}

    \RightLabel{$\LI L$}
    \BinaryInfC{$\Gamma ;\Delta_1,\Delta_3,P_3\LI P_4\SQ P_1$}

    \AxiomC{$\Xi$}
    \noLine
    \UnaryInfC{$\Gamma ; \Delta_2, P_2 \SQ A$}

  \RightLabel{$\LI L$}
  \BinaryInfC{$\Gamma ; \Delta_1,\Delta_2, \Delta_3, P_1 \LI P_2, P_3 \LI P_4 \SQ A$}
\end{customprooftree}
\end{center}
The remaining permutations involving non-simple $\LI$ and $\I$ left
introduction instances and follow this permutation closely.
\end{proof}

Instances of $absorb$ may appear anywhere prior to the use of its
principal formula. Without further meta-theoretical results, natural semantics style
derivation mimicry in Lolli will be modulo $absorb$ instance
placement. The following definition prescribes an exact placement for
all $absorb$ instance.

\begin{defn}[Coincided Proof]\label{def:coincidedproofs}
A coincided proof is a simple proof where every $absorb$ rule
instance unbounded premise formula corresponding to the principal
formula is the principal formula of a left introduction or identity
rule instance directly above it.
\end{defn}

The first proof in this section is not a coincided one because an
$absorb$ instance is detached where the proof in Subsection
\ref{subsec:speclogic} is a coincided one because all $absorb$ instances
satisfy Definition \ref{def:coincidedproofs}.

\begin{thm}[Lolli Admits Conincided Provability]\label{thm:coincidedprovability}
The sequent $\ubsequent{G}$ has a simple proof in Lolli if and only if
it has a conincided proof in Lolli.
\end{thm}

\begin{proof}
This proof is similar to the previous ones.  Observe that all
coincided proofs are simple ones, this satisfies the ``if''
direction. Now consider the ``only if'' direction. It is easy to see
that $absorb$ instances may be permuted up until they coincide with a
left introduction or identity rule instance. From this, we may
conclude the argument by induction on the measure of non-coincided
$absorb$ rule instances.
\end{proof}

Theorems \ref{thm:uniformprovability}, \ref{thm:simpleprovability},
and \ref{thm:coincidedprovability} can be used to yield a reduced
proof system that admits only coincided proofs. To do so we first
inductively define a unary predicate $||P||$ where $P$ is a program
clause formula that captures a backward chaining proof search
strategy. The predicate takes a program clause formula as an argument
and returns a set of triples where the first, second, and third
projection is a set of goal formulas, a multiset of goal formulas, and
a program clause formula, respectively. Each triple represents
unbounded(the first projection) and bounded(the second projection)
proof obligations for some program clause formula. An unbounded proof
obligation is one that must be provided strictly from the unbounded
context and a bounded proof obligation is one that must be proved from
some portion of the bounded context. Let $||P||$ be the smallest set
such that:
\begin{enumerate}
 \item $\langle \emptyset , \emptyset , A \rangle \in ||A||$
 \item if $\langle \Gamma , \Delta , P_1\AMPOP P_2 \rangle \in ||P|| $
then both $\langle \Gamma, \Delta , P_1 \rangle \in ||P|| $ and
$\langle \Gamma, \Delta , P_2 \rangle \in ||P||$
 \item if $\langle \Gamma, \Delta , \A x. P \rangle \in ||P|| $ then,
for all closed terms $t$, $\langle \Gamma, \Delta , P[t/x] \rangle \in
||P||$
 \item if $\langle \Gamma , \Delta, P_1 \I P_2 \rangle \in ||P||$ then
$\langle \Gamma\cup P_1, \Delta , P_2 \rangle \in ||P||$
 \item if $\langle \Gamma , \Delta, P_1 \LI P_2 \rangle \in ||P||$
then $\langle \Gamma, \Delta\uplus P_1 , P_2 \rangle \in ||P||$
\end{enumerate}

Let our specification logic have all right introduction rules from
Figure \ref{fig:lollirules} and the back chaining rules given in
Figure \ref{fig:backchaining}.

\begin{figure*}
\begin{framed}

\begin{customprooftree}

  \AxiomC{$\usequent{\emptyset}{B_1}$}
  \AxiomC{$\hdots$\ \ $\usequent{\emptyset}{B_n}$ $\usequent{\Delta_1}{C_1}$ \ \ $\hdots$}
  \AxiomC{$\usequent{\Delta_m}{C_m}$}
  \RightLabel{$BC_u$}
  \TrinaryInfC{$\sequent{\Gamma,B}{\Delta_1,\hdots,\Delta_m}{A}$}

  \AxiomC{$\usequent{\emptyset}{B_1}$}
  \AxiomC{$\hdots$\ \ $\usequent{\emptyset}{B_n}$ $\usequent{\Delta_1}{C_1}$ \ \ $\hdots$}
  \AxiomC{$\usequent{\Delta_m}{C_m}$}
  \RightLabel{$BC_b$}
  \TrinaryInfC{$\usequent{\Delta_1,\hdots,\Delta_m,B}{A}$}

\alwaysNoLine
\BinaryInfC{}
\end{customprooftree}

\caption{
In the specification logic, these back chaining rules will replaces all left-introduction rules from Figure \ref{fig:lollirules}.
Both have the proviso that $n,m\geq 0$ and $\langle \{ B_1,\hdots,B_n\},\{C_1,\hdots,C_m\},A\rangle \in ||B||$
}
\label{fig:backchaining}
\end{framed}
\end{figure*}

There are two forms of backward chaining in this figure both having as
their principal formula $B$. Intuitively, an instance of both could
replace a series of left introduction instances in a coincided
proof. Using the former requires that the left introduction series
begin (in a bottom-up reading) with an $absorb$ instance.

\begin{thm}[The Specification Logic and Lolli
Equivalence]\label{thm:logicsystemequivalence} The sequent
$\ubsequent{G}$ has a proof in Lolli if and only if it has a proof in
the specification logic.
\end{thm}

\begin{proof}
In the ``if'' direction, due to the definition of the backward
chaining rules, any back chaining instance in the specification logic
proof can be replace by some sequence of left introduction and
$absorb$ instances from Lolli.

Now, consider the ``only if'' direction.  Application of Theorems
\ref{thm:uniformprovability} followed by \ref{thm:simpleprovability}
and finally \ref{thm:coincidedprovability} allows us to convert a
Lolli proof to a coincided proof.  Finally, by Definition
\ref{def:simpleproofs} and the definition of $||P||$, we may replace
runs of left-introduction and $absorb$ instances by one of the two
instances of backward chaining.
\end{proof}

\section{Modeling Imperative Programming Languages}
\label{sec:formalization}

In this section, the imperative programming language defined in
Section \ref{sec:imp} is modeled using the specification logic
presented in Section \ref{sec:speclogic}. Additionally, proof mimicry
of derivations is demonstrated by considering the proof of a modeled
evaluation relation and that evaluation relations derivation.
Throughout this section and the rest of this paper, we refer to the
imperative programming language and its evaluation semantics as the
``object system''.

\subsection{The Model}
\label{subsec:formalization}

Our model extends the kinds of types we may have. Types in our model
will now include a type for programs in the model, $\programs$ and for
syntax representing natural numbers, $\naturals$ (again, overloaded
for use in the specification logic).

Let $\termmappingname$ be a function that translates programs from
${\cal L}$ given in Section \ref{sec:imp} into terms of type $\programs$ in the
specification logic.
\begin{equation}\label{eqn:termmapping}
\termmapping{i} =
\left\{
\begin{array}{ll}
  \vt{\NT{i}}  & \mbox{if } i \in \naturals \\
  \addt{\termmapping{j}}{\termmapping{k}}  & \mbox{if } i = j \hat{+} k \\
  \subt{\termmapping{j}}{\termmapping{k}}  & \mbox{if } i = j \hat{-} k \\
  \gtt{\termmapping{j}}{\termmapping{k}} & \mbox{if } i = j\hat{>}k \\
  \gett{\termmapping{j}} & \mbox{if } i = \hat{*}j \\
  \sett{\termmapping{j}}{\termmapping{k}} & \mbox{if } i = j\assign k \\  
  \seqt{\termmapping{j}}{\termmapping{k}} & \mbox{if } i = j;k \\
  \wht{\termmapping{j}}{\termmapping{k}} & \mbox{if } i = \mbox{while } j \mbox{ do } k
\end{array}
\right.
\end{equation}

Observe that any function can be represented as set of tuples relating
inputs to outputs. Such representations are often referred to as
function graphs. Let $\memmappingname$ be a recursive function that
translates memory function graphs to multisets composed exclusively of
occurrences of the binary predicate $\memformname$ with the type
$\naturals\rightarrow\naturals\rightarrow\omicron$. The first argument
to $\memformname$ represents a memory location and the second argument
represents the value stored at that location.
\begin{equation}\label{eqn:memmapping}
\memmapping{M} =
\left\{
\begin{array}{ll}
  \emptyset & \mbox{if } M = \emptyset\\
  \memform{l}{v} \uplus \memmapping{M'} & \mbox{if } M = {(l,v)} \cup M'
\end{array}
\right.
\end{equation}

The ternary evaluation predicate $\evaltname$ is defined in Figure
\ref{fig:evalt} and has the type
$\programs\rightarrow\naturals\rightarrow\omicron\rightarrow\omicron$.
This definition models the natural semantics rules given in Figure
\ref{fig:sosimp}.  Its first argument is the program expression to be
evaluated, its second argument is an element from $\naturals$
representing the return value of the input program, and its third
argument is a formula that must be proved in the memory left behind
after the program expression has been evaluated.  In this definition,
explicit quantification has been removed for clarity.  All capitalized
terms occurring in the head of a program clause formula are
universally quantified variables.  All capitalized terms occurring
exclusively in the body of a program clause formula are existentially
quantified variables.

The $\evaltname$ predicate relies heavily on a continuation-passing
style\cite{sussman75unp} where the universally quantified variable $C$
with type $\omicron$ is a continuation. The use of continuations
allows a natural way to express the subsequent evaluation of program
expression in potentially modified memory. For example, consider the
program clause formula in Figure \ref{fig:evalt} modeling sequencing
in the object system. As noted at the end of Subsection
\ref{subsec:evalsems}, one method for building a derivation would be
by building derivations for the premises in a left-to-right order. We
capture this method in this formula: the first program expression
should be evaluated and this may result in modified memory, the second
program expression should be modeled for evaluation in this modified
memory. Therefore, we extend the continuation with an evaluation
predicate for the second program expression.

For each natural semantics style rule given in section \ref{subsec:evalsems} there
is a corresponding formula in Figure \ref{fig:evalt}.  A simple
heuristic was followed to model each rule: modeled premises of a rule
``extend'' the continuation, becoming the body of a program clause
formula. Continuation extension is done in left-to-right premise
order. Finally, the conclusion of the rule will become the head of a
program clause formula.

As we did in Section \ref{sec:imp}, we will overloaded the operators
$+$, $-$, $>$, and $\leq$. Again, we associate the usual subtraction
operation on natural numbers: $N_1 - N_2$ is $0$ if $N_1$ is less than
$N_2$.
\begin{figure*}
\begin{framed}
\begin{center}

\begin{equation*}
\begin{array}{lll}
 C & \LI & \evalt{\vt{N}}{N}{C}\\
  \evalt{E_1}{N_1}{\evalt{E_2}{N_2}{(N_3 = N_1+N_2 \OT C)}} & \LI & \evalt{\addt{E_1}{E_2}}{N_3}{C}\\  \evalt{E_1}{N_1}{\evalt{E_2}{N_2}{(N_3 = N_1-N_2\OT C)}} & \LI & \evalt{\subt{E_1}{E_2}}{N_3}{C}\\
  \evalt{E_1}{N_1}{\evalt{E_2}{N_2}{(N_1>N_2\OT C)}} & \LI & \evalt{\gtt{E_1}{E_2}}{sz}{C}\\
  \evalt{E_1}{N_1}{\evalt{E_2}{N_2}{(N_1\leq N_2\OT C)}} & \LI & \evalt{\gtt{E_1}{E_2}}{z}{C}\\
  \evalt{E}{N_1}{(\memform{N_1}{N_2}\OT(\memform{N_1}{N_2}\LI C))} & \LI & \evalt{\gett{E}}{N_2}{C}\\
  \evalt{E_1}{N_1}{\evalt{E_2}{N_2}{(\memform{N_1}{N_3}\OT(\memform{N_1}{N_2}\LI C))}} & \LI & \evalt{\sett{E_1}{E_2}}{N_2}{C}\\
  \evalt{E_1}{N_1}{\evalt{E_2}{N_2}{C}} & \LI & \evalt{\seqt{E_1}{E_2}}{N_2}{C}\\
  \evalt{E_1}{N_1}{\evalt{E_2}{N_2}{\evalt{\wht{E_1}{E_2}}{C}}}\OT N_1> z & \LI & \evalt{\wht{E_1}{E_2}}{z}{C}\\
  \evalt{E_1}{N_1}{C}\OT N_1= z & \LI & \evalt{\wht{E_1}{E_2}}{z}{C}\\
\end{array}\end{equation*}

\end{center}
\caption{
The program clause formulas modeling the evaluation semantics given in
section \ref{subsec:evalsems} in the specification logic.
}
\label{fig:evalt}
\end{framed}
\end{figure*}

Let $\Gamma$ be a set exclusively containing the formulas from Figure
\ref{fig:evalt}. The evaluation relation $\evalrelation{E}{M}{N}{M'}$
defined in Section \ref{subsec:evalsems} is translated to the sequent:
\begin{equation}\label{eqn:evalmapping}
\usequent{\memmapping{M}}{\evalt{\termmapping{E}}{\NT{N}}{\top}}
\end{equation} The use of $\top$ here ``throws away'' the memory
resulting from this evaluation, i.e what is $M'$ in the evaluation
relation.  If inspection of this memory is necessary, we may replace
$\top$ with a goal formula. For example, if we wanted to inspect the
value in memory stored at location $1$ we could use the following
sequent:
\[
\usequent{\memmapping{M}}{\evalt{\termmapping{E}}{\NT{N_1}}{(\memform{1}{N_2}\OT\top)}}. \]

In our model of the object system, memory is accessed and modified
using the subformula
\[ \memform{N_1}{N_2}\OT(\memform{N_1}{N_3}\LI C) \] where $N_1,N_2,$
and $N_3$ have the type $\naturals$. When $N_2 = N_3$ the operation is
a lookup, otherwise it is an update. If $M$ is memory undefined at $N_1$
then a proof of this subformula must have the following structure due
to the meta-theoretical results from Subsection \ref{subsec:speclogic} and our
model.
\begin{customprooftree}\label{prooftree:nup}
    \AxiomC{$$}
    \RightLabel{$id$}
    \UnaryInfC{$\usequent{\memform{N_1}{N_2}}{\memform{N_1}{N_2}}$}

    \AxiomC{$\vdots$}
    \UnaryInfC{$\usequent{\memmapping{M},\memform{N_1}{N_3}}{C}$}

    \RightLabel{$\LI R$}
    \UnaryInfC{$\usequent{\memmapping{M}}{\memform{N_1}{N_3}\LI C}$}

  \RightLabel{$\OT R$}
  \BinaryInfC{$\usequent{\memmapping{M},\memform{N_1}{N_2}}{\memform{N_1}{N_2}\OT(\memform{N_1}{N_3}\LI C)}$}
\end{customprooftree}
Therefore, our treatment of state in our specification logic has an
intuitive and logical reading of ``remove the value $N_2$ at location $N_1$
and replace it with the value $N_3$''.

\subsection{Proofs as Computations}
\label{subsec:computations}

Consider the $\Omega$ derivation given in Subsection
\ref{subsec:derivations} for the relation
\[ \evalrelation{\ensuremath{1 \assign \deref{2}}}{M''}{M''(2)}{M'''}\] where
\[ M''=\memoryupdate{\memoryupdate{M}{2}{M(0)}}{0}{M(1)}\] and $M$ is
memory defined at $0,1$ and $2$. The sequent corresponding to this
evaluation relation is
\begin{tabbing} \qquad\=\qquad\=\kill
\>$\Gamma;\memform{\NT{0}}{\NT{M(1)}},\memform{\NT{1}}{\NT{M(1)}},\memform{\NT{2}}{\NT{M(0)}},\memmapping{O}\SQ$\\
\>\>$\evalt{(\ensuremath{1 \assign \deref{2}})}{M''(0)}{\top}$
\end{tabbing} where $O$ is memory and for all $n\in\naturals$ if $n>3$
then $O(n) = M(n)$, otherwise $O(n)$ is undefined.  A proof of this
sequent can be found in Figure \ref{fig:specproofex}.  In this proof
right introduction rules are omitted.

\begin{figure*}
\begin{framed}
\begin{customprooftree}
\AxiomC{}
\RightLabel{$\top R$}
\UnaryInfC{$\vdots$}
\noLine
\UnaryInfC{$\usequent{\memform{\NT{0}}{\NT{M(1)}},\memform{\NT{1}}{\NT{M(1)}},\memform{\NT{2}}{\NT{M(0)}},\memmapping{O}}{(\memform{\NT{2}}{\NT{M(0)}} \OT \memform{\NT{2}}{\NT{M(0)}}) \LI (\memform{\NT{1}}{\NT{M(1)}} \OT (\memform{\NT{1}}{\NT{M(0)}} \LI \top))}$}
\RightLabel{$BC_U$}
\UnaryInfC{$\usequent{\memform{\NT{0}}{\NT{M(1)}},\memform{\NT{1}}{\NT{M(1)}},\memform{\NT{2}}{\NT{M(0)}},\memmapping{O}}{\evalt{\NT{2}}{\NT{2}}{(\memform{\NT{2}}{\NT{M(0)}} \OT (\memform{\NT{2}}{\NT{M(0)}} \LI (\memform{\NT{1}}{\NT{M(1)}} \OT (\memform{\NT{1}}{\NT{M(0)}} \LI \top))))}}$}
\RightLabel{$BC_U$}
\UnaryInfC{$\usequent{\memform{\NT{0}}{\NT{M(1)}},\memform{\NT{1}}{\NT{M(1)}},\memform{\NT{2}}{\NT{M(0)}},\memmapping{O}}{\evalt{\gett{\NT{2}}}{\NT{M(0)}}{(\memform{\NT{1}}{\NT{M(1)}} \OT (\memform{\NT{1}}{\NT{M(0)}} \LI \top))}}$}
\RightLabel{$BC_U$}
\UnaryInfC{$\usequent{\memform{\NT{0}}{\NT{M(1)}},\memform{\NT{1}}{\NT{M(1)}},\memform{\NT{2}}{\NT{M(0)}},\memmapping{O}}{\evalt{\NT{1}}{\NT{1}}{\evalt{\gett{\NT{2}}}{\NT{M(0)}}{(\memform{\NT{1}}{\NT{M(1)}} \OT (\memform{\NT{1}}{\NT{M(0)}} \LI \top))}}}$}
\RightLabel{$BC_U$}
\UnaryInfC{$\usequent{\memform{\NT{0}}{\NT{M(1)}},\memform{\NT{1}}{\NT{M(1)}},\memform{\NT{2}}{\NT{M(0)}},\memmapping{O}}{\evalt{\sett{\NT{1}}{\gett{\NT{2}}}}{\NT{M(0)}}{\top}}$}
\end{customprooftree}
\caption[]{
A proof of the sequent $\usequent{\memform{\NT{0}}{\NT{M(1)}},\memform{\NT{1}}{\NT{M(1)}},\memform{\NT{2}}{\NT{M(0)}},\memmapping{O}}{\evalt{(\ensuremath{1 \assign \deref{2}})}{M''(0)}{\top}}$.
}
\label{fig:specproofex}
\end{framed}
\end{figure*}
A mimicry of the derivation can be seen in this proof: for every rule
instance that occurs in the derivation there is a corresponding $BC_u$
instance with a principal program clause formula1 that models that
derivation rule instance.
\section{Reasoning about Properties of Imperative 
Programs Using the Model}
\label{sec:reasoning}
In this section, we show that our model of the object system can be
used to prove a similar property to what was shown in Subsection
\ref{subsec:sosreasoning}. As it was in Subsection
\ref{subsec:sosreasoning}, the property proven is trivial. However,
the goal of this exercise is to demonstrate that the structure of the
argument on the model follows very closely the structure of the
argument from Subsection \ref{subsec:sosreasoning}. In this sense,
reasoning about properties of our model can be intuitive. This
advantage when reasoning is a result of the mimicry exposed in
Subsection \ref{subsec:computations}.

\subsection{Correctness as a Property of Proofs in the Model}
\label{subsec:correctness} 
We must model the Lemma \ref{lem:sumcorrectnesssoslemma} and Theorem
\ref{thm:sumcorrectnesssostheorem} in the specification logic.  The
modeled lemma and theorem rely on the sum program from Equation
\ref{eqn:sum}, the term translation function from Equation
\ref{eqn:termmapping}, the memory translation function from Equation
\ref{eqn:memmapping}, the evaluation predicate $\evaltname$ from
Figure \ref{fig:evalt}, and , tacitly, the relationship translation
function from Equation \ref{eqn:evalmapping}.  As defined in
Subsection \ref{subsec:formalization}, the set $\Gamma$ is exclusively
inhabited by formulas from Figure \ref{fig:evalt}.

Lemma \ref{lem:sumcorrectnessformlemma} encodes Lemma
\ref{lem:sumcorrectnesssoslemma} from Subsection
\ref{subsec:sosreasoning} in the specification logic.

\begin{lem}[Total Correctness of
$\termmapping{Q}$]\label{lem:sumcorrectnessformlemma} $\A N_1, N_2, M,
M_0$ if $N_1,N_2\in\naturals$ and $M =
\memoryupdate{\memoryupdate{M_0}{\NT{0}}{\NT{N_2}}}{\NT{1}}{\NT{N_1}}$
then $\E N_3,$ $N_3\in\naturals$ and
$\usequent{\memmapping{M}}{\evalt{\termmapping{Q}}{\NT{0}}{\memform{\NT{0}}{\NT{N_3}}\OT\top}}$
and $N_3 = N_2+\sum\limits_{i = 0}^{N_1} i$
\end{lem}

The value of a memory location was extracted by function application
in the object system.  In the encoding we retrieve the value from
memory after program evaluation via a continuation formula.
Specifically, in Lemma \ref{lem:sumcorrectnessformlemma} that
continuation formula is $\memform{\NT{0}}{\NT{N_3}}\OT\top$.  This
formula extracts only the value in memory at location $0$.  This is
where we expect the result of program $Q$ to be stored.

The encoding of Theorem \ref{thm:sumcorrectnesssostheorem} is similar.
Observe that $N_1$ and $N_2$ are immediately initialized upon
evaluation of $P$; this use of $N_1$ and $N_2$ is only meant to ensure
that memory $M$ is defined at locations $0$ and $1$.

\begin{thm}[Total Correctness of
$\termmapping{P}$]\label{thm:sumcorrectnessformtheorem} $\A N_1, N_2,
M, M_0$ if $N_1,N_2\in\naturals$ and $M =
\memoryupdate{\memoryupdate{M_0}{\NT{0}}{\NT{N_2}}}{\NT{1}}{\NT{N_1}}$
then $\E N_3,$ $N_3\in\naturals$ and
$\usequent{\memmapping{M}}{\evalt{\termmapping{P}}{\NT{0}}{\memform{\NT{0}}{\NT{N_3}}\OT\top}}$
and $N_3 = \sum\limits_{i = 0}^{N} i$
\end{thm}
 
\subsection{Reasoning about Proofs}
\label{subsec:reasoning}

Reasoning about the model is structured according to the
reasoning structure in Subsection \ref{subsec:sosreasoning}.

\begin{proof}[Proof of Lemma \ref{lem:sumcorrectnessformlemma}] 
This will be shown by induction on $N_1$. In the first case where
$N_1=0$, the proof in Figure \ref{fig:baseevalproof} can be
constructed and we can conclude that $N_3 = N_2$.  Therefore, we have
that $N_3 = N_2+\sum\limits_{i = 0}^{N_1} i$.
\begin{figure*}
\begin{framed}
\begin{customprooftree} \input{baseevalproof}
\end{customprooftree}
\caption[]{ A proof of the judgment
$\usequent{\memmapping{M}}{\evalt{\termmapping{Q}}{\NT{0}}{\memform{\NT{0}}{\NT{N_3}}\OT\top}}$
where $N_1,N_2\in\naturals$, $N_1=0$, $M =
\memoryupdate{\memoryupdate{M_0}{\NT{0}}{\NT{N_2}}}{\NT{1}}{\NT{N_1}}$,
and $M'$ is a partial function undefined at $0$,$1$ and equal to $M_o$
otherwise.  }
\label{fig:baseevalproof}
\end{framed}
\end{figure*}

In the second case we assume that Lemma
\ref{lem:sumcorrectnessformlemma} holds if $M(1)<N_1$ and must show
this lemma holds when $M(1) = N_1$; this is our inductive hypothesis.
Proof analysis of the sequent
\begin{tabbing}
\qquad\=\qquad\=\kill
\>$\Gamma;\memmapping{M'},\memform{\NT{1}}{\NT{M(1)}},\memform{\NT{0}}{\NT{M(0)}}\SQ$\\
\>\>$\evalt{\termmapping{Q}}{\NT{0}}{(\memform{\NT{0}}{\NT{N_3}}}$
\end{tabbing}
 reveals that it suffices to build a proof for the sequent
\begin{tabbing}
\qquad\=\qquad\=\kill
\>$\Gamma;\memform{\NT{1}}{\NT{(M(1)-1)}},\memform{\NT{0}}{\NT{(M(0)+M(1))}},\memmapping{M'}$\\
\>\>$\SQ\evalt{\termmapping{Q}}{\NT{0}}{(\memform{\NT{0}}{\NT{N_3}}\OT\top)}.$
\end{tabbing}
We omit such a proof in our discussion here; it mimics the derivation
for the second case given in Subsection \ref{subsec:sosreasoning}, it
is tedious, and, its construction is completely mechanizable in our
specification logic. The inductive hypothesis yields this sequent.
Additionally, by the inductive hypothesis, we have $N_4 =
(M(0)+M(1))+\sum\limits_{i = 0}^{M(1)-1} i$ for some
$N_4\in\naturals$.  This is equivalent to $N_4 = M(0)+\sum\limits_{i =
0}^{M(1)} i$ and thus, $N_3 = N_4$.
\end{proof}

\begin{proof}[Proof of Theorem \ref{thm:sumcorrectnessformtheorem}]

We must prove
\[
\usequent{\memmapping{M}}{\evalt{\termmapping{P}}{\NT{0}}{\memform{\NT{0}}{\NT{N_3}}\OT\top}}\]
and $N_3 = \sum\limits_{i = 0}^{N} i$. Proof analysis of this sequent
reveals it is sufficient to prove the sequent
\[
\usequent{\memmapping{M'},\memform{\NT{1}}{N},\memform{\NT{0}}{0}}{\evalt{\termmapping{Q}}{\NT{0}}{(\memform{\NT{0}}{\NT{N_3}}
\OT \top)}}\]. We have both by Lemma
\ref{lem:sumcorrectnessformlemma}.
\end{proof}

\subsection{Extracting Properties from the Model}
\label{subsec:extraction}
We would like to extract Lemma \ref{lem:sumcorrectnessformlemma} and
Theorem \ref{thm:sumcorrectnessformtheorem} into the object system.
In general, doing so requires some confidence that the extracted
property is meaningful in the object system.  Such confidence is
typically acquired through an informal adequacy argument
\cite{harper07jfp}.

The adequacy of an encoding can be shown by giving a bijective
translation function from the object system to the encoding.  There
are complexities in providing such a translation for our encoding; in
the object system, memory is a term while in the encoding memory is
formula.  How such a translation can be given is left to future work.
\section{Conclusion}
\label{sec:conclusion}

We have considered in this paper the possibility of formalizing the
process of reasoning about properties of imperative programs. 
Towards this end, we have described a specification logic that can
transparently model imperative programming languages with semantics
defined in an natural semantics style. 
An important aspect of this specification logic is that its proof
relation can be restructured so as to yield derivations that closely
resemble the ones that may be constructed in the original natural semantics style
encodings of object systems. 
We have illustrated how this characteristic can be exploited in
reasoning about the properties of the object systems.
In our example, we have used an informal style of reasoning over
specification logic derivations.
However, we believe that this reasoning process can be formalizing and
we are examining this aspect in ongoing work. 
In particular, we are exploring the idea of using a two-level logic
approach~\cite{gacek12jar,mcdowell02tocl} that has been
successfully exploited in conjunction with an intuitionistic
specification logic in the Abella system~\cite{gacek12jar}.
In this approach, we encode a specification logic via its derivability
relation within a rich ``reasoning'' logic: by using the capabilities
of the reasoning logic, we then obtain the ability to prove properties
about derivations in the specification logic.
One of our immediate goals is to accommodate a linear specification
logic within the same reasoning logic that underlies Abella, thereby
producing a variant of Abella that supports the development of formal
arguments related to systems oriented around resource usage. 
Once we have an implementation of such a system at hand, the next step
would be to use it to formalize the kinds of arguments we have
presented in this paper. 

In addition to actually implementing the ideas we have discussed in
this paper within a formal system, we must also extend them so that we
can reason about a larger, more realistic collection of programs. 
The imperative programs that we have considered in this paper use
programming language constructs permitting non-termination and memory
manipulations, 
i.e. lookup and update. In essence, we have demonstrated that our
approach can be effective when reasoning about properties of basic
imperative 
programs lacking pointers (because memory values were never used in
lookups) or dynamic allocation.  
Going forward, we would like to examine two particular kinds of
extensions to this work.

\begin{quote}
  \textbf{The language} chosen in this paper does not permit complex
  notions of data, dynamic memory allocation, or functional aspects.
  It does permit references but the imperative program analyzed does
  not use them. A more relevant language to model would be a subset of SML
  \cite{milner90sml} excluding data-type definitions and the module
  subsystem. This subset would not make modeling evaluation semantics
  much more complex. For example, memory allocation can be treated
  naturally using universal quantifiers. We conjecture that such
  changes will not alter the intuitive nature of reasoning.

  \textbf{The program} chosen and its correctness property is
  trivial. Programs in common use among other researchers concerned
  with reasoning about imperative programs are linked list (singly or
  doubly) manipulation programs and implementations of the
  Schorr-Waite algorithm\cite{bornat00mpc}. Additionally, properties
  of programs using references can be particularly difficult to reason
  about due to aliasing. Aliasing occurs when a location can be
  accessed in two different ways. For example, the program
  \[ \ensuremath{1 \assign 0 ; 2 \assign 1 ; 3 \assign 1 ; \deref{2} \assign 4 ; \deref{\deref{3}}} \] is one where aliasing
  occurs; the last two program expressions will update and lookup,
  respectively, location $1$. Reasoning in basic Hoare logic is
  unsound when programs containing aliasing are considered. Hoare
  logic can be extended such that reasoning about pointers is
  technically feasible but complex\cite{ohearn01csl}. We have an
  inchoate idea that a treatment of aliasing should not require major
  changes to our specification logic because locations are not
  named. Therefore, references to aliased data is explicit. How this
  treatment will affect reasoning intuitions remains unclear.
\end{quote}

Finally, we must better understand the connections between our
approach and others such as ones using Hoare and separation based
logic\cite{ohearn01csl,brochenin07lfcs,mehta03cade,jia06pls}, pointer
assertion logic\cite{jensen97pldi}, parametric shape
analysis\cite{levami00sa}, and aliasing logic\cite{bozga04sa}.
\bibliographystyle{abbrvnat}
\bibliography{local}

\end{document}